\documentclass[conference]{IEEEtran}
\usepackage{cite}
\usepackage{amsmath,amssymb,amsfonts}
\usepackage{algorithmic}
\usepackage{graphicx}
\usepackage{textcomp}
\usepackage{xcolor}
\def\BibTeX{{\rm B\kern-.05em{\sc i\kern-.025em b}\kern-.08em
    T\kern-.1667em\lower.7ex\hbox{E}\kern-.125emX}}

\usepackage{comment}
\usepackage[hyphens]{url}
\usepackage{amssymb}
\usepackage{pifont}
\newcommand{\cmark}{\ding{51}}%
\newcommand{\xmark}{\ding{55}}%
\newtheorem{definition}{Definition}
\newtheorem{theorem}{Theorem}
\newtheorem{proof}{Proof}

\begin{document}

\title{Privacy Analysis and Evaluation Policy of Blockchain-based Anonymous Cryptocurrencies}

\author{
\IEEEauthorblockN{Takeshi Miyamae}
\IEEEauthorblockA{\textit{Security Laboratory} \\
\textit{Fujitsu Laboratories Ltd.}\\
Kawasaki, Japan \\
miyamae.takeshi@fujitsu.com}
\and
\IEEEauthorblockN{Kanta Matsuura}
\IEEEauthorblockA{\textit{Institute of Industrial Science} \\
\textit{The University of Tokyo}\\
Tokyo, Japan \\
kanta@iis.u-tokyo.ac.jp}
}

\maketitle

\begin{abstract}
In blockchain-based anonymous cryptocurrencies, due to their tamper-resistance and transparency characteristics, transaction data are initially required to be anonymous, with the help of various cryptographic techniques, e.g., commitment schemes and zero-knowledge proofs.
Also, cryptocurrencies are different from existing anonymous messaging protocols regarding the software architecture and the underlying security model.
Due to these differences, the sense of anonymity must be specifically defined for anonymous cryptocurrencies, and the anonymity in each anonymous cryptocurrency must be analyzed and evaluated based on the specific architecture model.

In this paper, we first propose a specific architecture model with three software layers to anonymous cryptocurrencies.
Next, we introduce definitions of fundamental privacy properties (Pfitzmann's anonymity, unlinkability, and pseudonymity) and comprehensively analyze each privacy property for each architecture layer of anonymous cryptocurrencies to establish a privacy evaluation policy for anonymous cryptocurrencies.
Finally, we fairly compare the privacy of current leading anonymous cryptocurrencies (e.g., Zerocash, CryptoNote, and Mimblewimble) using the privacy evaluation policy.
\end{abstract}


\section{Introduction}
\label{sec:introduction}

\subsection{Background}
Since blockchain-based technologies usually involve trustless security models, it is essential to consider privacy. Particularly in Bitcoin \cite{Nakamoto2008}, quite a few privacy concerns due to its transparency of the transaction details on the public ledger have been highlighted. For example, an attacker attempts to infer from the transaction graph that two or more seemingly unrelated Bitcoin addresses belong to a Bitcoin user (address linking attack) \cite{Ron2013}. On the other hand, another attacker attempts to link Bitcoin users with their real-world user identities from the indirect behavior of the Bitcoin transactions (de-anonymization) \cite{Reid2013}.

However, privacy-enhancing techniques for blockchain-based anonymous cryptocurrencies (after this referred to as ``anonymous cryptocurrencies'') have not been sufficiently compared and analyzed so far because anonymous cryptocurrencies are relatively recent and unique transaction systems. Unlike the transaction data accumulated in-house, where the data are first collected and protected in each company and can be disclosed to other companies for data analysis, transaction data on a blockchain are assumed to be shared by even malicious users from the beginning. In the case of in-house databases, the data are typically assumed to be modified with generalization, suppression, or substitution techniques, including local and global recoding \cite{Terrovitis2011} and microaggregation \cite{DomingoFerrer2002}, to achieve anonymity when disclosed. And these techniques can be evaluated based on specific privacy criteria such as k-anonymity \cite{Sweeney2002}\cite{Sweeney2002a} or l-diversity \cite{Machanavajjhala2007}. In the case of blockchain-based transaction methods, however, from the characteristics of tamper-resistance and transparency, the data are initially required to be anonymous, with the help of cryptographic techniques such as commitment schemes and zero-knowledge proofs. Therefore, it is unusual for anonymous cryptocurrencies to benefit from such modification techniques that can be applied to in-house databases.

Moreover, while most anonymous messaging protocols \cite{Unger2015} are also initially anonymous, the underlying security model for anonymous cryptocurrencies is different. Besides, the software architecture of anonymous cryptocurrencies is more complicated than that of anonymous messaging protocols.
Due to these differences, the sense of anonymity in anonymous cryptocurrencies must be specifically analyzed, and the anonymity in each anonymous cryptocurrency must be evaluated based on the analysis.

To analyze and formalize the anonymity in anonymous cryptocurrencies, we also have to use supplemental measures of anonymity. For example, Wu \cite{Wu2018} advocates that
\begin{equation}
anonymity = pseudonymity + unlinkability,
\end{equation}
referring to Narayanan's study \cite{Narayanan2016}. Although this formula is intuitively correct, the logical relations between each of the measures, e.g., one is a necessary or sufficient condition of another, are not evident from the old works in the field of anonymous cryptocurrency.

In this paper, we first propose an architecture model for anonymous cryptocurrency with three software layers (Section \ref{sec:architecture_model}).
Next, we introduce the definitions of fundamental privacy properties (Pfitzmann's anonymity, unlinkability, and pseudonymity) (Section \ref{sec:definitions_privacy_properties}), and comprehensively analyze each privacy property for each architecture layer in anonymous cryptocurrencies that we defined in this paper (Section \ref{sec:analysis_privacy_properties}).
In this analysis, we propose four linkability attack models (SLLA, TLLA, RCCLA, SCCLA), with which we cover most of the privacy vulnerabilities in anonymous cryptocurrencies.

\subsection{Related Works}
Androulaki et al. \cite{Androulaki2013} proposed a progressive definition of ``address unlinkability'' in Bitcoin and a way to quantify it. However, our main interest is the software architecture of anonymous cryptocurrency, including a secret-sharing model between a sender and a recipient that makes privacy analysis complicated or items of interest (IOIs) of anonymous cryptocurrency that can suffer from privacy attacks. Therefore, both do not conflict.

Although academic researchers in the field of anonymous cryptocurrencies have already noticed that unlinkability is highly essential, the definition of unlinkability depends on the researchers, which makes it difficult for us to understand the exact meaning of unlinkability and the relationship between anonymity and unlinkability, etc.
For example, the author of CryptoNote \cite{Saberhagen2013} defined unlinkability as ``for any two outgoing transactions it is impossible to prove they were sent to the same person'', which meant only a specific aspect of Pfitzmann's unlinkability \cite{Pfitzmann2010}. In contrast, the authors of Zerocash \cite{BenSasson2014} introduced another privacy property called `ledger indistinguishability,' which enhances more robust privacy than unlinkability. Therefore, Zerocash implicitly achieves unlinkability without using the word `unlinkability.'

No one opposes that \cite{Bonneau2015} is one of the great surveys of blockchain technology. However, since only a few pages can be spared for the topic of privacy, the discussion of unlinkability is not sufficient. \cite{Amarasinghe2019} is also an excellent survey, but the criteria for evaluating unlinkability are not so elaborated and fairly applied to each anonymous cryptocurrency. As a result, in all the cases of Zerocash, CryptoNote, and Mimblewimble, the level of unlinkability is ``Moderately High'', which does not help anything for us after all.

\section{Architecture Model for Anonymous Cryptocurrency}
\label{sec:architecture_model}

In the case of Bitcoin, which is always transparent like other non-anonymous cryptocurrencies, all the transaction information is disclosed to anyone.
Therefore, a sender does not know about the ledger state more than other users can know at all.
On the other hand, usually in the case of anonymous cryptocurrencies, only the sender and its corresponding recipient know about the sent coins.
Conversely, if other users could accidentally share anything about the coins with the sender and recipient for some reason, it would never be an excellent anonymous cryptocurrency.

Figure \ref{fig:architecture_of_anonymous_cryptocurrency} depicts a general software architecture of anonymous cryptocurrency.
We introduce a software model for anonymous cryptocurrency that comprises transport layer (Layer-0), ledger layer (Layer-1), and secret-sharing layer (Layer-2).
As shown in \cite{Gudgeon2020}, blockchain applications with three layers is not unusual.
Although Layer-2 usually means payment-channel protocols, however, calling the secret-sharing layer in anonymous cryptocurrencies `Layer-2' is our novel viewpoint.

From the viewpoint of data repository, an anonymous cryptocurrencies generally depend on each user's local secure database (usually called `wallet') to allow the sender and recipient to share some secret of sent coins, whereas non-anonymous cryptocurrencies such as Bitcoin only use their distributed ledger (usually called `blockchain').

\begin{figure}
\centering
\includegraphics[width=7.7cm]{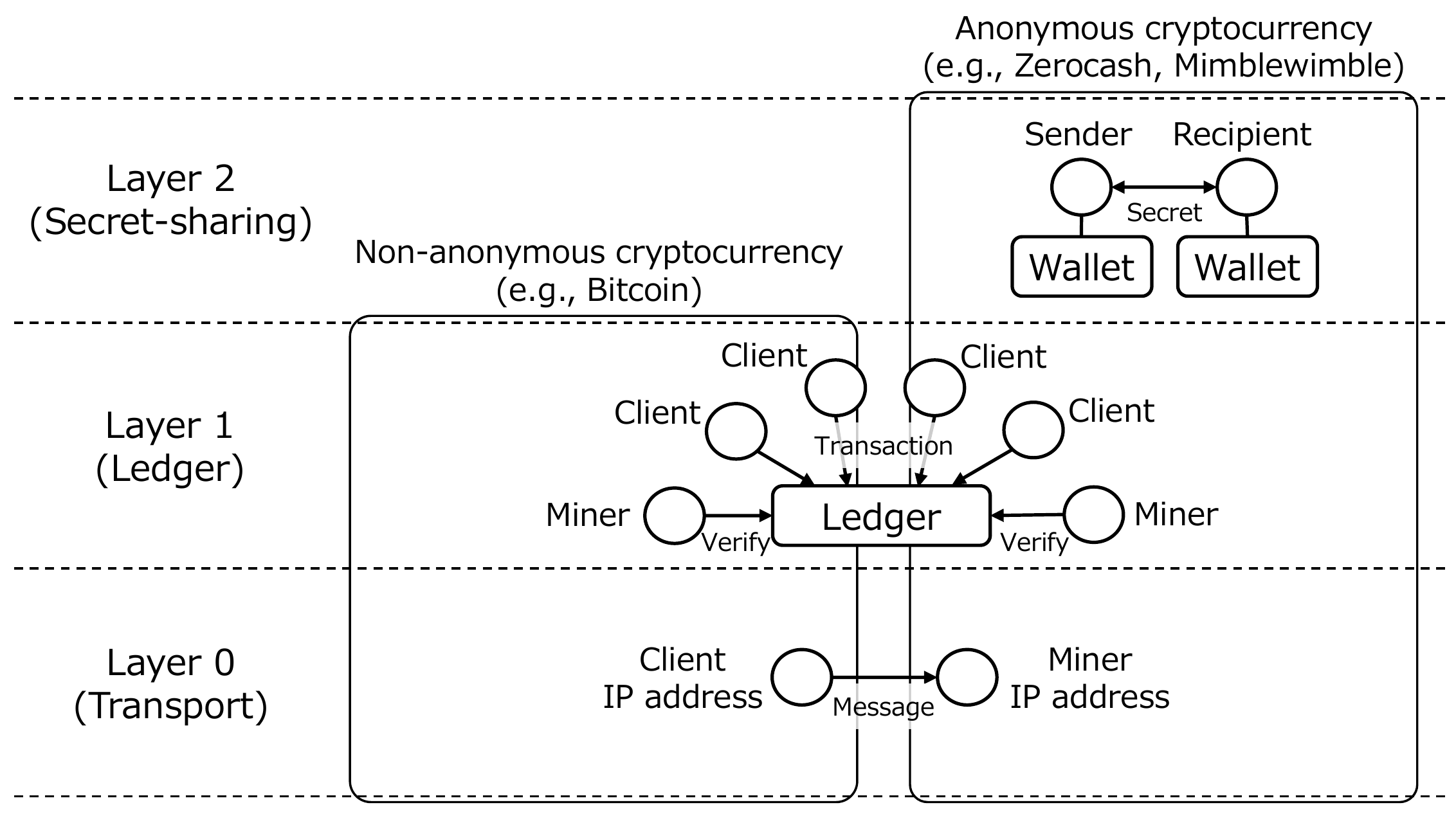}
\caption{Architecture Model for Anonymous Cryptocurrency}
\label{fig:architecture_of_anonymous_cryptocurrency}
\end{figure}

\section{Definitions of Fundamental Privacy Properties}
\label{sec:definitions_privacy_properties}

In this section, we introduce the definitions of fundamental privacy properties (Pfitzmann's anonymity, unlinkability, and pseudonymity) \cite{Pfitzmann2010}.

\subsection{Anonymity}
\label{subsub:anonymity}
Pfitzmann defines anonymity in a highly general fashion for researchers who are engaged in any information and communication systems to analyze privacy concerns.

\begin{definition}
[Pfitzmann's Anonymity] Anonymity of a subject from an attacker's perspective means that the attacker cannot sufficiently identify the subject within a set of subjects, the anonymity set.
\end{definition}

First, Pfitzmann's privacy properties are defined using the abstract words of entities, such as `subjects' (senders) `execute actions' (send using a communication network) `on objects' (messages) to `subjects' (recipients), as depicted in Figure \ref{fig:pfitzmann_anonymity}.

\begin{figure}
\centering
\includegraphics[width=7.7cm]{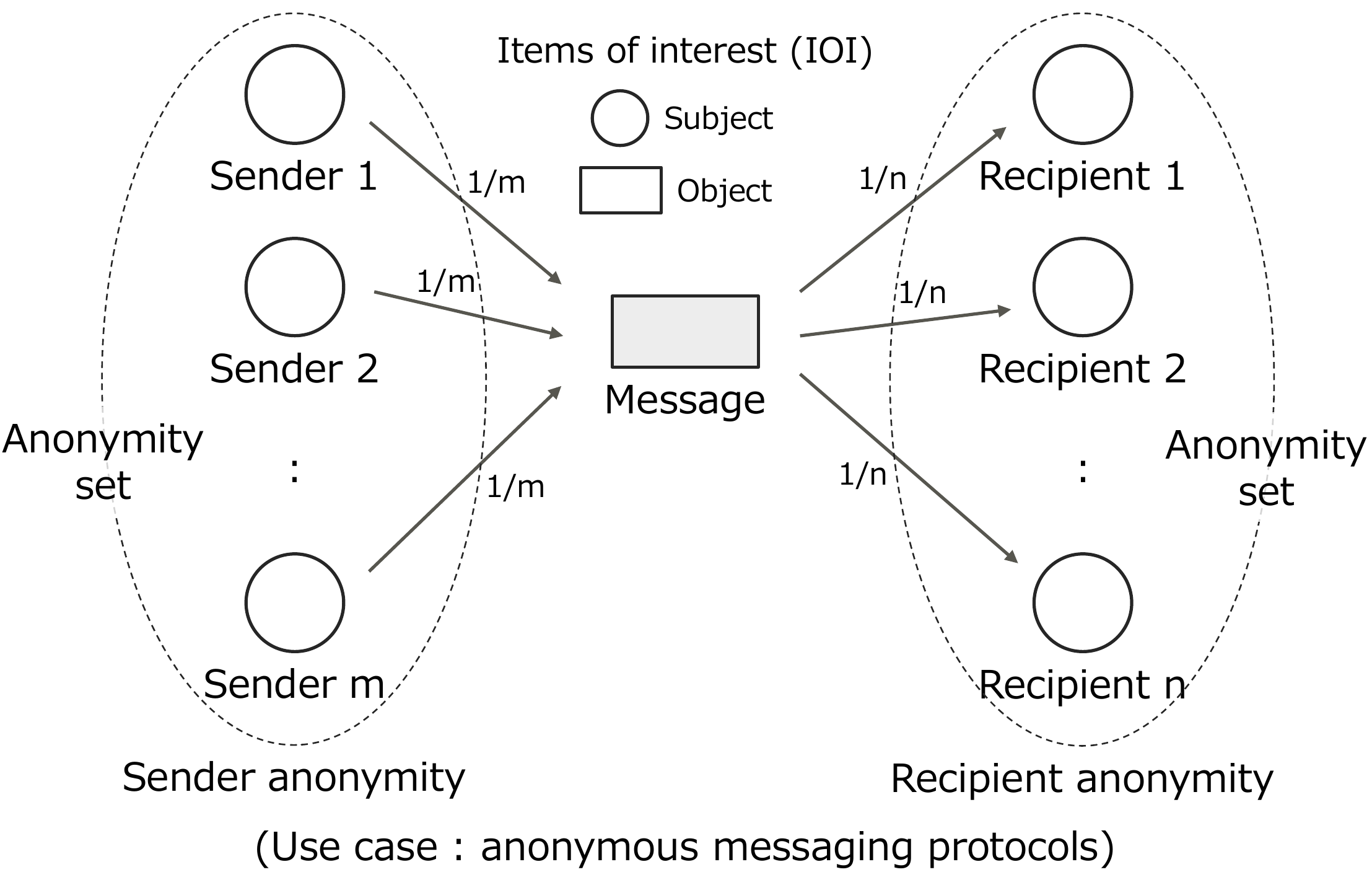}
\caption{Pfitzmann's Definition of Anonymity}
\label{fig:pfitzmann_anonymity}
\end{figure}

Second, the `anonymity set' is a set of all possible subjects among which the subject is identified. Regarding senders in a messaging protocol, for example, the anonymity set consists of the subjects that could send a message. Since anonymity sets could differ over time, a sender can only be anonymous within a set of potential senders at each time, which may be a subset of worldwide all subjects that can send a message from time to time.

Third, the word ``sufficiently'' makes it possible for people to quantify the anonymity.

Forth, note that each statement is made from the perspective of an attacker who may be interested in monitoring what kind of communication is ongoing, or what kind of messages are communicated.

Finally, we will mention the capabilities of an attacker against anonymity.
\begin{itemize}
\item An attacker against anonymity does not forget anything if they are once informed. Therefore, the anonymity set can never increase.
\item An attacker against anonymity can manipulate the communication message if they need to.
\item An attacker against anonymity can be an outsider tapping communication lines, while they can also be an insider capable of engaging in regular communications.
\end{itemize}

Finally, by comparing it with Common Criteria's definition of anonymity \cite{CC2012}, we can find other advantages of Pfitzmann's definition.

\begin{definition}
[Common Criteria's Anonymity] This family (the functional family of anonymity) ensures that a user may use a resource or service without disclosing the user's identity. The requirements for anonymity provide protection of the user identity. Anonymity is not intended to protect the subject identity. [...] Anonymity, requires that other users or subjects are unable to determine the identity of a user bound to a subject or operation.
\end{definition}

In the case of Common Criteria, the definition of ``identity'' seems independent of the definition of anonymity; therefore, we are required to define ``identity'' using, e.g., personally identifiable information (PII) \cite{NIST2010}, to complete the definition of anonymity.
In Pfitzmann's definition, on the other hand, the sense of ``identify the subject'' can be shown by defining the specific cases of anonymity defined in \ref{subsub:anonymity_in_terms_of_unlinkability}.

\subsection{Unlinkability}
\label{subsub:unlinkability}
Pfitzmann defines unlinkability as indistinguishability of items of interest (IOIs).

\begin{definition}
[Pfitzmann's Unlinkability]
\label{def:unlinkability}
Unlinkability of two or more items of interest (IOIs, e.g., subjects, messages, actions, ...) from an attacker’s perspective means that within the system (comprising these and possibly other items), the attacker cannot sufficiently distinguish whether these IOIs are related or not.
\end{definition}

One of the differences between the definitions of anonymity and unlinkability is that anonymity is applied to a single entity, while unlinkability is applied to multiple entities.
Another is that anonymity can only be applied to a subject, while unlinkability can be applied to any entity.

We can intuitively understand Definition \ref{def:unlinkability} by citing another part from the Pfitzmann's paper \cite{Pfitzmann2010}.

\begin{quote}
\normalsize
In a scenario with at least two senders (a sender anonymity set), two messages sent by subjects within the same anonymity set are unlinkable for an attacker if, for him, the probability that these two messages are sent by the same sender is sufficiently close to 1/(the number of senders).
\end{quote}

Figure \ref{fig:pfitzmann_unlinkability} illustrates the claim in this citation (we call this scenario ``message unlinkability'').

\begin{figure}[h]
\centering
\includegraphics[width=7.7cm]{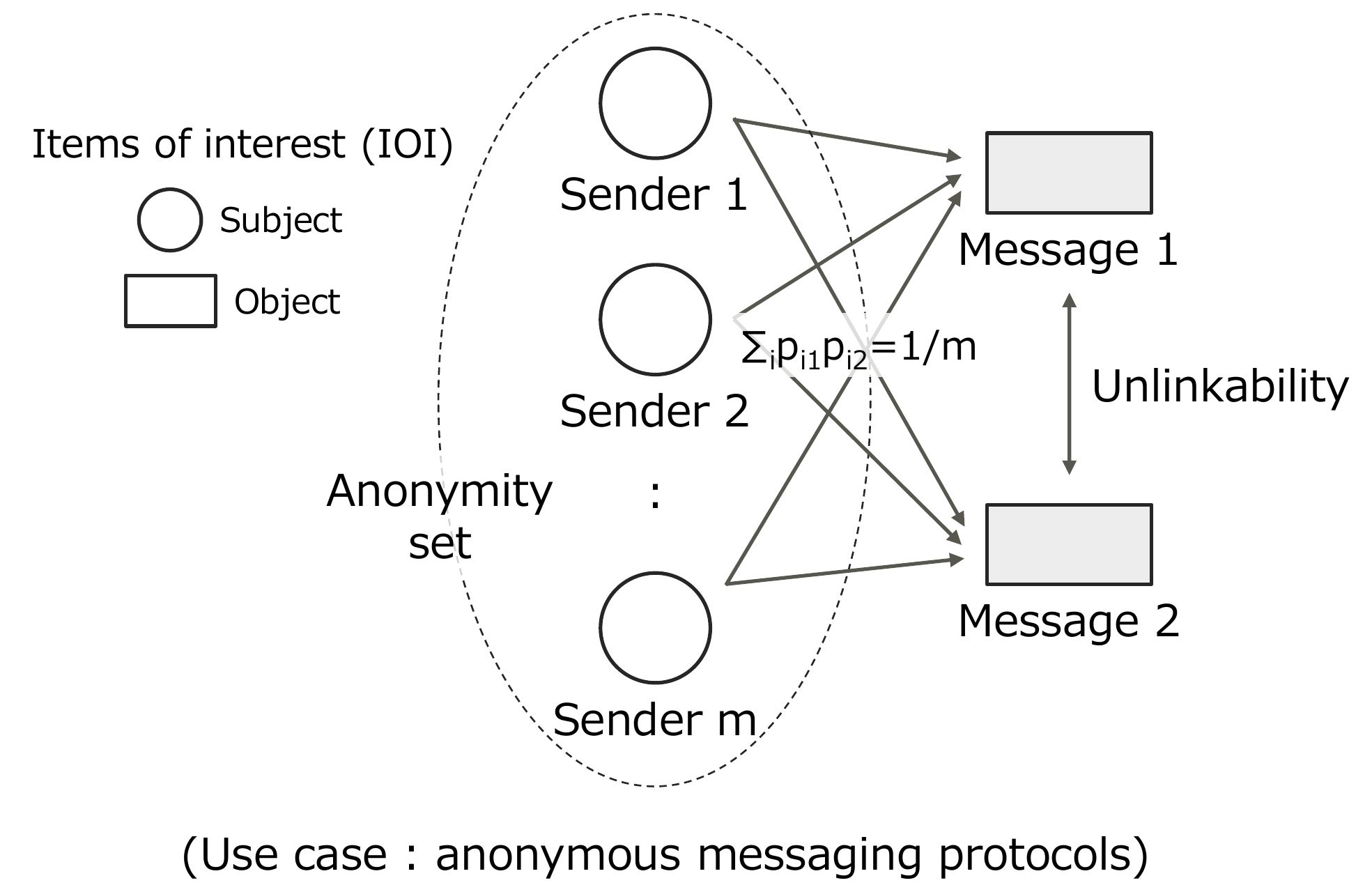}
\caption{Pfitzmann's Definition of Unlinkability}
\label{fig:pfitzmann_unlinkability}
\end{figure}

The indistinguishability in Definition \ref{def:unlinkability} is expressed in more accessible terms such as ``not sent by the same sender.'' This expression suggests how to apply the definition to real-world privacy problems.

We can extend the scenario to the case of more than two messages, where total unlinkability is achieved if and only if any two of the messages achieve a primitive unlinkability.
Moreover, we can apply the definition of unlinkability to different IOIs, e.g., a sender and recipient, by the medium of another IOI, e.g., a message.

To easily extend the application of unlinkability, the following theorems we found are convenient (Unlinkability is a partial equivalence relation).

\begin{theorem}
The relation of unlinkability between two sets of IOIs is symmetric.
\end{theorem}

\begin{proof}
Let $A$ be an IOI set, where the number of elements is $n_A$, and $e_A \in A$ ($B$, $n_B$, $e_B$, etc. are similarly defined). If $B$ achieves unlinkability by the medium of $A$, then the probability that a randomly picked-up element in $B$ is $e_B$ and the associated element in $A$ with $e_B$ is $e_A$ is $1/{n_A}{n_B}$. Hence, the probability of $e_B$ under the condition of $e_A$ is ${(1/{n_A}{n_B})}/{(1/{n_A})} = 1/{n_B}$. Therefore, $A$ achieves unlinkability by the medium of $B$.
\end{proof}

\begin{theorem}
The relation of unlinkability between two sets of IOIs is transitive.
\end{theorem}

\begin{proof}
If $A$ achieves unlinkability by the medium of $B$, and $B$ achieves unlinkability by the medium of $C$, then the probability under the condition that the associated element in $B$ with $e_a$ is $e_b$ and the associated element in $C$ with $e_b$ is $e_c$ is $1/{n_b}{n_c}$. As the probability of $e_b$ and that of $e_c$ is independent, the probability under the condition that the associated element in $C$ with $e_a$ is $e_c$ is $\sum_{B} 1/{n_b}{n_c} = 1/{n_c}$. Therefore, $A$ achieves unlinkability by the medium of $C$.
\end{proof}

Finally, Pfitzmann defines linkability as an antonym of unlinkability.

\begin{definition}
[Pfitzmann's Linkability]
Linkability of two or more items of interest (IOIs, e.g., subjects, messages, actions, ...) from an attacker’s perspective means that within the system (comprising these and possibly other items), the attacker can sufficiently distinguish whether these IOIs are related or not.
\end{definition}

\subsection{Relationship between Anonymity and Unlinkability}
\label{sub:relationship_anonymity_unlinkability}

\subsubsection{Anonymity in terms of Unlinkability}
\label{subsub:anonymity_in_terms_of_unlinkability}

Pfitzmann defines the varieties of anonymity in the specific scenario of anonymous messaging protocols.

\begin{definition}
[Sender Anonymity]
\label{def:sender_anonymity}
Given a message sent from someone, we call ``sender anonymity'' as a specific case of anonymity, in which the possibility, from an attacker's perspective, that the message is sent from each sender is 1/(the number of senders).
\end{definition}

\begin{definition}
[Recipient Anonymity]
Given a message sent to someone, we call ``recipient anonymity'' as a specific case of anonymity, in which the possibility, from an attacker's perspective, that the message is sent to each recipient is 1/(the number of recipients).
\end{definition}

\begin{definition}
[Relationship Anonymity]
Given a message sent from someone to someone, we call ``relationship anonymity'' as a specific case of anonymity, in which the possibility, from an attacker's perspective, that the message is sent from each sender to each recipient is 1/(the number of sender and recipient pairs).
\end{definition}

Note that relationship anonymity can be defined from the viewpoint of outsiders only, i.e., an attacker being neither a sender nor a recipient of a message.

\subsubsection{Equivalency of Sender Anonymity and Message Unlinkability}

Given two or more messages and a sender anonymity set, the message unlinkability defined in \ref{subsub:unlinkability} is a necessary and sufficient condition for the sender anonymity defined in Definition \ref{def:sender_anonymity} in \ref{subsub:anonymity}.

If $a_\alpha$ is a probability that the sender of the messages $\alpha$ is Alice (other variables are similarly defined) and that the senders of both two messages are anonymous (sender anonymity), the two messages are unlinkable (message unlinkability).
\begin{center}
$a_\alpha=b_\alpha=c_\alpha=1/3, a_\beta=b_\beta=c_\beta=1/3$
$\Rightarrow a_\alpha a_\beta+b_\alpha b_\beta+c_\alpha c_\beta=1/3$
\end{center}
Inversely, if the two messages are unlinkable (message unlinkability), the senders of both messages are anonymous (sender anonymity).
\begin{center}
$a_\alpha, b_\alpha, c_\alpha, a_\beta, b_\beta, c_\beta \geq 0,$
$a_\alpha+b_\alpha+c_\alpha=1, a_\beta+b_\beta+c_\beta=1,$
$a_\alpha a_\beta+b_\alpha b_\beta+c_\alpha c_\beta=1/3$
$\Rightarrow a_\alpha=b_\alpha=c_\alpha=1/3, a_\beta=b_\beta=c_\beta=1/3$
\end{center}

Likewise, recipient anonymity and relationship anonymity can be shown as equivalent to message unlinkability as well.

From the fact that all the cases in which anonymity can be defined are included in those in which unlinkability can be defined, unlinkability in all cases is a sufficient condition of anonymity. Inversely, since unlinkability is not a necessary condition of anonymity, eliminating unlinkability between some farther IOIs from the subjects we are interested in does not necessarily eliminate the anonymity of the subjects.

\subsection{Pseudonymity}
Pfitzmann defines pseudonym and pseudonymity in the following manner.

\begin{definition}
[Pfitzmann's Pseudonym] A pseudonym is an identifier of a subject other than one of the subject's real names.
\end{definition}

\begin{definition}
[Pfitzmann's Pseudonymity] Pseudonymity is the use of pseudonyms as identifiers.
\end{definition}

Pseudonymity is a more general notion of a pseudonym.
It is defined as the process of preparing for the use of pseudonyms, e.g., by establishing specific rules on how and under which conditions civil identities of holders of pseudonyms will be disclosed by so-called identity brokers.
An identity broker is a particular type of certification authority for pseudonyms.
Since anonymity can be described as a particular type of unlinkability, cf. \ref{sub:relationship_anonymity_unlinkability}, the concept of identity broker can be generalized to a linkability broker.

\subsection{Untraceability}
In the research field of cryptocurrencies, ``untraceability'' is often used with the same meaning as ``unlinkability''. In particular, untraceability has often been used in papers on e-cash systems \cite{Okamoto1992}\cite{Asokan1997} before the appearance of blockchains and cryptocurrencies. However, since untraceability is known as just a particular type of unlinkability, as Wu mentioned \cite{Wu2018}, we do not treat untraceability as an independent privacy property in this paper.

\section{Analysis of Privacy Properties in Anonymous Cryptocurrency}
\label{sec:analysis_privacy_properties}

In this section, we comprehensively analyze each privacy property for each architecture layer in anonymous cryptocurrencies that we defined in this paper.
In this analysis, we propose four linkability attack models (SLLA, TLLA, RCCLA, SCCLA), with which we cover most of the privacy vulnerabilities in anonymous cryptocurrencies.

Since Pfitzmann's definitions of anonymity and unlinkability \cite{Pfitzmann2010} are most general and most comfortable to understand as far as we know, we adopt it for analyzing anonymous cryptocurrencies.

However, as cryptocurrencies have specific unique characteristics, anonymity has to be considered differently from anonymous messaging protocols, which Pfitzmann chooses as a typical application instance of his theory.

The main difference is the layer of software implementation. In the case of anonymous messaging protocols, the content of each message is handled as just opaque data.
In the case of cryptocurrencies, however, transaction details (e.g., the sender of coins, the recipient of coins, and the amount of each coin) are described in each message that a sender sends to a recipient.
Therefore, privacy attacks in anonymous cryptocurrencies can be made from each layer of Figure \ref{fig:architecture_of_anonymous_cryptocurrency}, while anonymous messaging protocols only need to protect the attacks in a transport layer.

\subsection{Ledger Layer (Layer-1)}
\subsubsection{IOI Model in Layer-1}
In the data structure of cryptocurrencies, several attributes, e.g., sender, recipient, coin, amount of coin, transaction, and confirmation time, can generally be considered as IOIs, as shown in IOI model in Layer-1 of anonymous cryptocurrency (Figure \ref{fig:anonymous_cryptocurrency_ioi_model}). Among them, only sender and recipient are categorized as subjects, which are qualified to be argued regarding anonymity. Therefore, considering the unlinkability of coins by the medium of senders or recipients is the most critical to achieving sender or recipient anonymity. If both sender and recipient anonymity are achieved, we can easily show that relationship anonymity between the sender and recipient is also achieved (\ref{sub:relationship_anonymity_unlinkability}). On the other hand, if either sender anonymity or recipient anonymity is not achieved, we have to conclude that coins are linkable. Therefore, we are required to analyze the unlinkability of other cryptocurrency attributes in the IOI model to identify to what extent the unlinkability of the anonymous cryptocurrency is achieved.

\begin{figure}
\centering
\includegraphics[width=7.7cm]{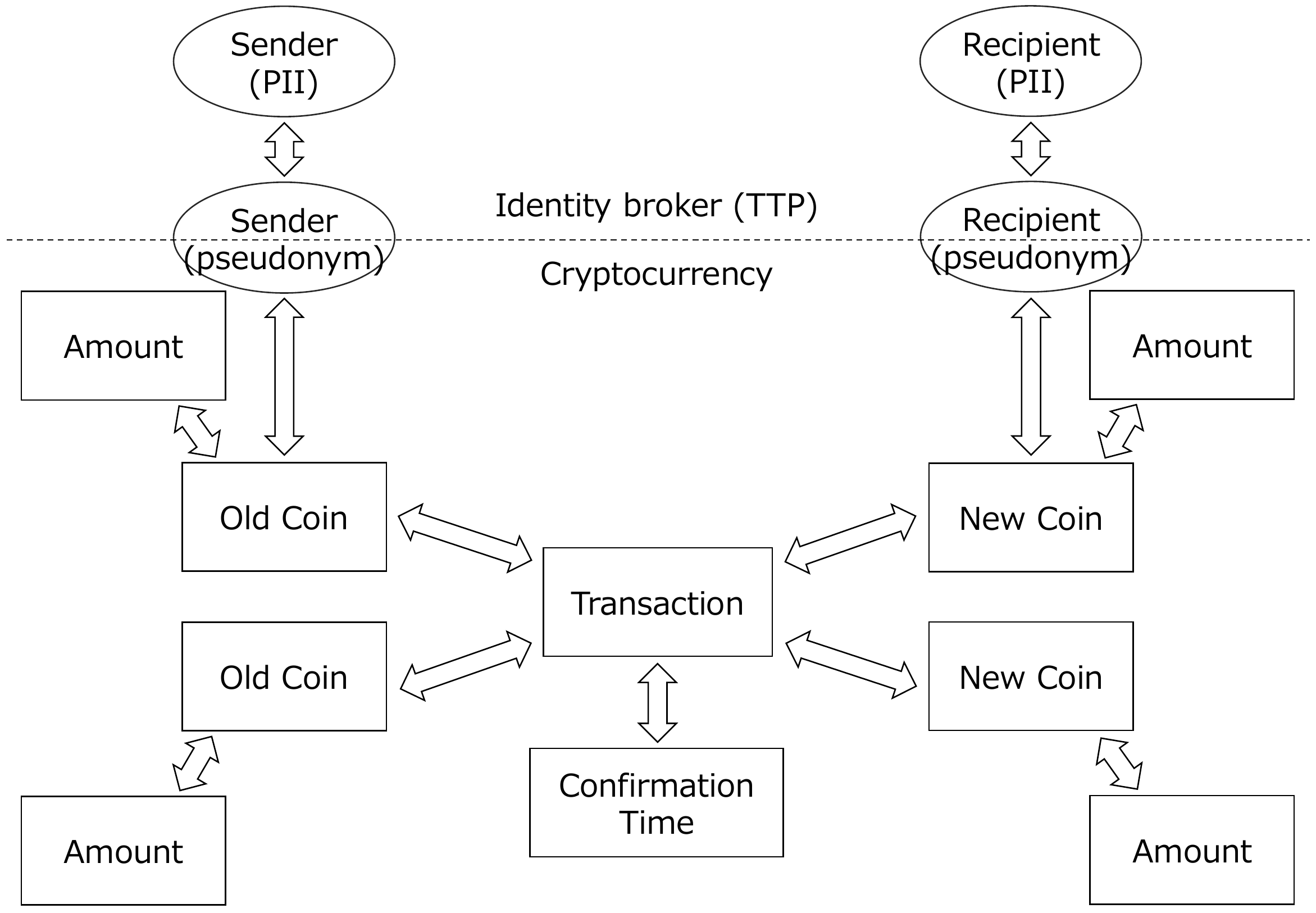}
\caption{IOI Model in Layer-1 of Anonymous Cryptocurrency}
\label{fig:anonymous_cryptocurrency_ioi_model}
\end{figure}

\subsubsection{Pseudonymity}
\label{subsub:pseudonymity_in_layer1}
In most anonymous cryptocurrencies, digital signature schemes based on public-key cryptography are used so that a coin sender can prove that they are the correct recipient of each consumed coin. As the public key is typically a random number, it is conveniently used as a base of each user's pseudonym, which is also called ``(cryptocurrency) address,'' to enhance its privacy instead of using an identification number that suggests a particular user's PII.

Pseudonymization undoubtedly contributes to the anonymity in Layer-1 of anonymous cryptocurrencies.
As shown in IOI model in Layer-1 of anonymous cryptocurrency (Figure \ref{fig:anonymous_cryptocurrency_ioi_model}), any two pseudonyms achieve unlinkability by the medium of the sender's PII, and at the same time by the medium of the recipient's PII.
From the viewpoint of anonymous cryptocurrencies, it seems that PII management is conveniently outsourced to the linkability brokers, e.g., cryptocurrency exchanges.

However, in reality, pseudonymization is useful only when the linkability brokers can securely protect the link between each pseudonym and its corresponding PII.
Therefore, whether anonymity can be achieved in anonymous cryptocurrencies depends on the conventional cybersecurity techniques that current linkability brokers are typically using.
Moreover, since blockchain-based anonymous cryptocurrencies are usually designed to be decentralized, depending on the techniques managed by the TPP is inconsistent with the fundamental design policy of blockchain.

Another concern is the indirect de-anonymization attack using the transaction graph.
As current techniques of analyzing transaction graphs are highly advanced \cite{Meiklejohn2013}\cite{Reid2013}\cite{Androulaki2013}\cite{Goldfeder2014}, indirect de-anonymization is still a threat even if the pseudonym is ideally applied.
Therefore, pseudonymization is not a sufficient condition but just a necessary condition for anonymity.
For anonymous cryptocurrencies to achieve sufficient anonymity, unlinkability between IOIs is still required.

In the case of Coinshuffle \cite{Ruffing2014} and CryptoNote, the recipient address is always freshly created for each new transaction (one-timeness).
Owing to the one-time address, the sender and recipient can keep consumed coin and the newly created coin unlinkable, which makes indirect de-anonymization attacks more challenging.

\subsubsection{Unlinkability}
\label{subsub:unlinkability_in_layer1}
We will define a linkability attack model in which an attacker tries to obtain any information from a shared ledger (blockchain). In this attack model, any user can become an attacker.

\begin{definition}
Shared Ledger Linkability Attack (SLLA):\\
An attacker tries to obtain any unknown information about a victim only from a shared ledger to link with the information they already have on condition that the attacker cannot send any coins to or receive any coins from anyone, including the victim.
\end{definition}

Since shared ledgers are exposed to an unspecified large number of public users, the unlinkability between information on the ledger is required from the beginning. Therefore, SLLA-resistance is the most fundamental requirement for anonymous cryptocurrencies. Most anonymous cryptocurrencies encrypt each information on the shared ledger to achieve unlinkability.

Regarding coin-to-coin unlinkability, Zerocoin \cite{Miers2013} and Zerocash perfectly achieves it due to its another technical coin identifier called `serial number.'
Although a sender discloses a coin commitment as a first coin identifier when the sender sends a coin, the recipient newly declares a serial number as a substitute coin identifier when the recipient consumes the coin in the future.
Therefore, from the other users' perspective, since they can only see seemingly irrelevant two identifiers regarding the same coin, they cannot link those transactions and thus coin transfers.

In the case of mixing techniques, e.g., CoinJoin \cite{Maxwell2013} and Coinshuffle, or a one-time ring signature in CryptoNote, coins are only shuffled with a predefined number of decoys.
This fact means that if the shuffle is repeated over the multiple degrees of coin transfer, coin-to-coin unlinkability would be achieved exponentially with the number of degrees.
However, in the adjacent coin transfer, a link between a sender and a recipient would be probabilistically guessed by $1/n$ ($n$: the number of decoys).

In the case of Mimblewimble, since the same Pedersen commitment is used as a coin identifier, coins can always be linked; that is, anyone can construct a transaction graph.

Regarding coin-to-value unlinkability, it depends on whether each coin value is encrypted or not.

Regarding coin-to-time unlinkability, which means that attackers cannot know when a sender sends a coin or when the receiver consumes the coin, it depends on the transaction management of each anonymous cryptocurrency.
However, since in most of the current anonymous cryptocurrencies the transaction management is delegated to an unspecified large number of miners, coin-to-time unlinkability is usually not achieved.

SLLA-resistance of each cryptocurrency is scrutinized in Appendix \ref{sec:appendix_unlinkability_evaluation}.

\subsubsection{Anonymity}
As discussed in \ref{subsub:pseudonymity_in_layer1}, most of the anonymous cryptocurrencies use a randomized pseudonym.
Therefore, both sender and recipient anonymity are achieved at that point, and there is no way for the privacy attackers to succeed in direct de-anonymization.
Moreover, in the case of Mimblewimble, since cryptocurrency address itself is not used (addressless scheme), not to mention pseudonym, it can never suffer from direct de-anonymization attacks.

On the other hand, the resistance to the indirect de-anonymization attacks of an anonymous cryptocurrency is in the first place affected by its SLLA-resistance discussed in \ref{subsub:unlinkability_in_layer1} (the SLLA-resistance in this context focuses only on coin-to-coin unlinkability).
Even if an anonymous cryptocurrency is not SLLA-resistant, e.g., SLLA-unresistant or SLLA-probabilistically-resistant, however, privacy-enhancing address schemes allow it to achieve better anonymity, because such schemes make it difficult for indirect de-anonymization attackers to associate someone's IOI information by the medium of a single address of them.
That is, the condition for anonymous cryptocurrencies to be resistant to indirect de-anonymization is at least one of the following unlinkable property or addressing schemes.

\begin{enumerate}
\item SLLA-resistant
\item addressless scheme
\item address encryption
\item one-time address
\end{enumerate}

\subsection{Transport Layer (Layer-0)}
\subsubsection{Pseudonymity}
In cryptocurrencies, generally speaking, transport address (= IP address) is an identifier only from which it is not easy to calculate PII of coin sender or recipient.
However, since each transport address is not a random value (actually it usually includes physical or virtual segment information), it can provide a fatal geographical or ISP information that helps to de-anonymize the coin sender or recipient.

Moreover, an IP address once assigned to an active user would not be changed even if it is managed using DHCP, which results in being kept away from one-timeness.
The longer the assignment period, the higher the probability of inducing de-anonymization of the sender or recipient by the medium of their IP addresses.

\subsubsection{Unlinkability}
\label{subsub:layer0_unlinkability}
The unlinkability of the transport layer can be considered from each of Layer-0 or Layer-1 information perspective.

From the Layer-0 information perspective, since miners recognize a transaction as an opaque message, only transaction-to-transaction unlinkability is concerned.
However, it is easy for the miners to associate two or more transactions of the same sender by the medium of their transport address.

From the Layer-1 information perspective, we will define a linkability attack model by the medium of transport layer information.

\begin{definition}
Transport Layer Linkability Attack (TLLA):\\
An attacker, among transaction relay nodes including blockchain miners, tries to obtain any unknown information about a victim, who sends a transaction to a shared ledger or sends a message to someone for any other purpose, to link with the information they already have by the medium of transport layer information, e.g., IP address. Note that the information that can be retrieved from the ledger is included in the attacker's initial knowledge.
\end{definition}

In most cryptocurrencies, including Bitcoin \cite{Nakamoto2008}, miners are entrusted with verifying transactions sent by the users. If a miner in a cryptocurrency can see transaction details, however, they can link the transaction details with an IP address of the sender unless an anonymous messaging protocol is used. If the IP address is already associated with the sender's PII by the miner, the miner can link the transaction details with the PII by the medium of the IP address. That is, attackers using TLLA can easily achieve the linkability of transaction details and sometimes de-anonymize the transaction sender by the medium of IP address (cf. Figure \ref{fig:cryptocurrency_ioi_model_ipaddress}).

\begin{figure}
\centering
\includegraphics[width=7.7cm]{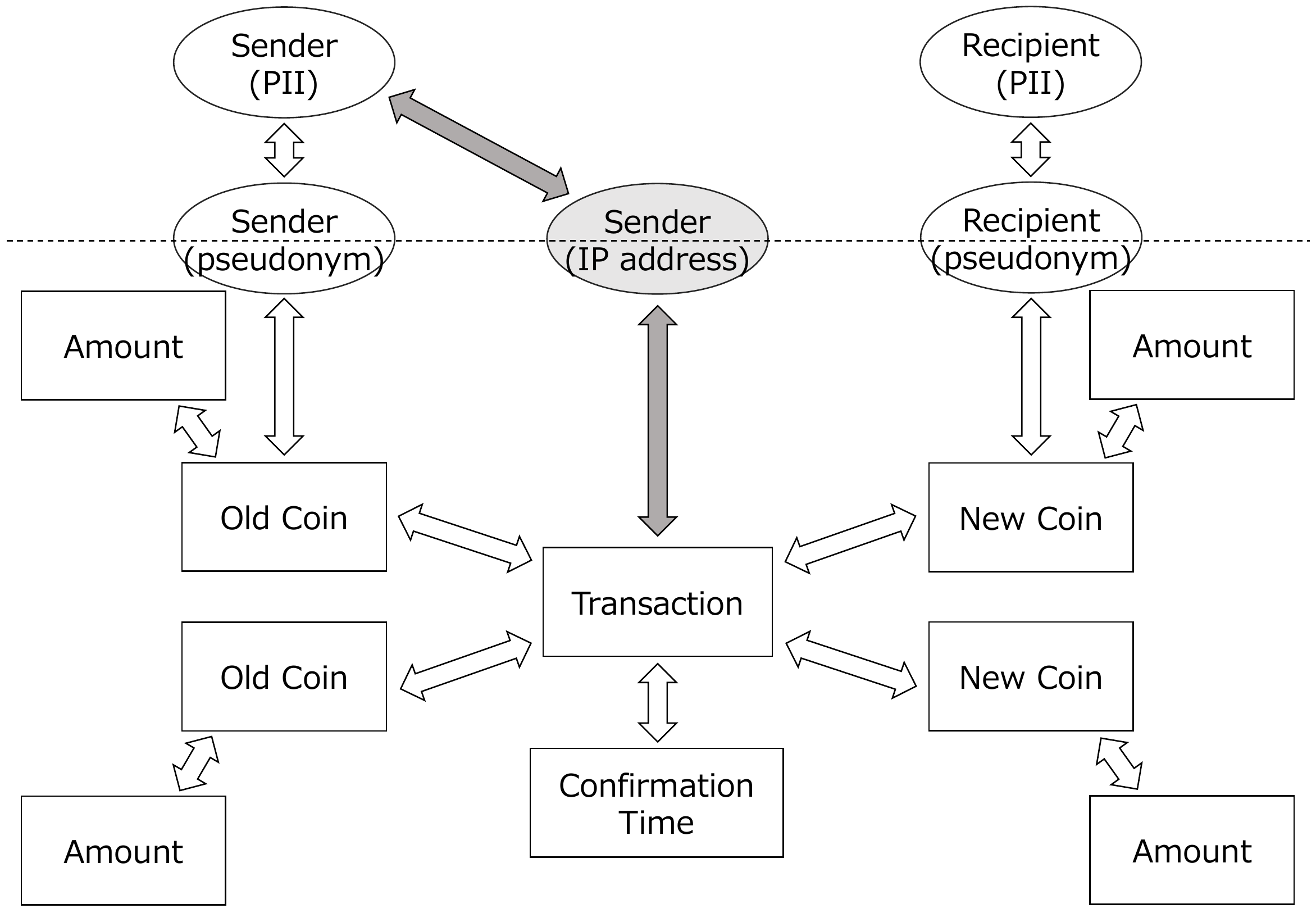}
\caption{IOI Model in Layer-1+0 of Anonymous Cryptocurrency}
\label{fig:cryptocurrency_ioi_model_ipaddress}
\end{figure}

In Zerocash \cite{BenSasson2014} and Mimblewimble \cite{Poelstra2016}, since all the information in each transaction, including the sender's address, is encrypted, the miners cannot obtain any unknown information about the sender from the transaction information. However, since the sender sends their transactions to the miner using their transport address, the miner can link two or more encrypted transactions by the medium of the transport address.

In the case of CryptoNote \cite{Saberhagen2013}, since encrypted information in a transaction is only stealth address, the miner can also link other transaction information, e.g., the amount of each coin, by the medium of the transport address.

Moreover, in the case of Mimblewimble, a sender of a coin has to send a commitment with their secrets $r$ and $v$ to a recipient using their transport address.
Therefore, the sender can link two or more coins sent to the same recipient by the medium of recipient’s transport address, and vice versa.

In the end, in anonymous cryptocurrencies, both from the perspective of Layer-0 and Layer-1 information, for us to prevent the linkability, transport addresses have to be anonymized using anonymization techniques in the transport layer such as Tor \cite{Dingledine2004}, DC-Nets \cite{Chaum1988}, etc.

\subsubsection{Anonymity}
In most anonymous cryptocurrencies, since coin senders send the transactions to the miners, miners know their transport addresses. (Mimblewimble is the only exception in which coin recipients send the transactions to the miners.)
Moreover, if a miner could share some association information between someone's transport address and PIIs for some reason, they would directly de-anonymize them.
For example, they can link activities of sending coins with the sender's PII by the medium of the transport address, e.g., whether the sender sends a particular transaction or not, or when the sender sends the transaction.
Although direct de-anonymization can be protected so long as linkability brokers of transport addresses are trusted and resistant to cyber-attacks, conventional cybersecurity techniques do not necessarily work so much as expected.

On the other hand, indirect de-anonymization from both the information perspective of Layer-0 can be considered in the same way as the indirect de-anonymization in Layer-1, because indirect de-anonymization does not use transport addresses as mediums.
Indirect de-anonymization can also be protected so long as all the information in each transaction is encrypted.
However, this depends on the privacy-enhancing techniques for anonymous cryptocurrencies, as shown in \ref{subsub:layer0_unlinkability}.

\begin{table*}[t]
\centering
\scalebox{0.55}[0.55]{
\begin{tabular}{c|c|c|c|c|c|c|c|c|c|c|c|c|c|c} \hline
& \multicolumn{6}{l|}{Cryptocurrency information} & \multicolumn{8}{l}{Privacy evaluation} \\ \hline
& \multicolumn{6}{l|}{} & \multicolumn{8}{l}{Ledger layer (Layer-1)} \\ \hline
& \multicolumn{6}{l|}{} & \multicolumn{2}{l|}{Pseudonymity} & \multicolumn{3}{l|}{Unlinkability} & \multicolumn{3}{l}{Anonymity} \\ \hline
& Year & Basic features & \multicolumn{4}{l|}{Privacy enhancing schemes} & \multicolumn{2}{l|}{} & \multicolumn{3}{l|}{SLLA-resist.} & \multicolumn{2}{l|}{Direct de-anon. resistance} & Indirect de-anon. resistance \\ \hline
& & Decentralization	& Secret-share & Addressless & Addr encrypt & Value encrypt & Randomness & One-timeness & Coin-to-coin & Coin-to-value & Coin-to-time & Sender anon. & Recip. anon. & \\ \hline
Bitcoin & 2008 & \cmark\cmark & \xmark & \xmark & \xmark & \xmark & \cmark\cmark & \xmark & \xmark & \xmark & \xmark & \cmark\cmark(pseudonym) & \cmark\cmark(pseudonym) & \xmark(SLLA-unresist.) \\ \hline
CoinJoin & 2013 & \xmark & \xmark & \xmark & \xmark & \xmark & \cmark\cmark & \xmark & \cmark(prob.) & \xmark & \xmark & \cmark\cmark(pseudonym) & \cmark\cmark(pseudonym) & \cmark(SLLA-prob.-resist.) \\ \hline
Coinshuffle & 2014 & \cmark\cmark & \xmark & \xmark & \xmark & \xmark & \cmark\cmark & \cmark\cmark & \cmark(prob.) & \xmark & \xmark & \cmark\cmark(pseudonym) & \cmark\cmark(pseudonym) & \cmark\cmark(SLLA-prob.-resist., one-time address) \\ \hline
Zerocoin & 2013 & \cmark(trusted setup) & \cmark\cmark & \cmark\cmark & - & \xmark & - & - & \cmark\cmark & \xmark & \xmark & \cmark\cmark(addressless) & \cmark\cmark(addressless) & \cmark\cmark(SLLA-resist., addressless) \\ \hline
Zerocash & 2014 & \cmark(trusted setup) & \cmark\cmark & \xmark & \cmark\cmark & \cmark\cmark & \cmark\cmark & \xmark & \cmark\cmark & \cmark\cmark & \xmark & \cmark\cmark(pseudonym) & \cmark\cmark(pseudonym) & \cmark\cmark(SLLA-resist., address encryption) \\ \hline
CryptoNote & 2013 & \cmark\cmark & \cmark\cmark & \xmark & \cmark\cmark & \cmark\cmark(Ring CT) & \cmark\cmark & \cmark\cmark & \cmark(prob.) & \cmark\cmark(Ring CT) & \xmark & \cmark\cmark(pseudonym) & \cmark\cmark(pseudonym) & \cmark\cmark(SLLA-prob.-resist., one-time address) \\ \hline
Mimblewimble & 2016 & \cmark\cmark & \cmark\cmark & \cmark\cmark & - & \cmark\cmark & - & - & \xmark & \cmark\cmark & \xmark & \cmark\cmark(addressless) & \cmark\cmark(addressless) & \cmark\cmark(SLLA-unresist., addressless) \\ \hline
\end{tabular}
}
\end{table*}

\begin{table*}
\centering
\scalebox{0.55}[0.55]{
\begin{tabular}{c|c|c|c|c|c|c|c|c|c|c|c|c|c} \hline
& \multicolumn{13}{l}{Privacy evaluation} \\ \hline
& \multicolumn{6}{l|}{Transport layer (Layer-0)} & \multicolumn{7}{l}{Secret-sharing layer (Layer-2)} \\ \hline
& \multicolumn{2}{l|}{Pseudonymity} & \multicolumn{2}{l|}{Unlinkability} & \multicolumn{2}{l|}{Anonymity} & \multicolumn{6}{l|}{Unlinkability} & Anonymity \\ \hline
& \multicolumn{2}{l|}{} & Layer-0 & TLLA-resist. & \multicolumn{2}{l|}{Direct de-anon. resistance} & \multicolumn{3}{l|}{RCCLA-resist.} & \multicolumn{3}{l|}{SCCLA-resist.} & Indirect de-anon. resistance \\ \hline
& Randomness & One-timeness & Tran-to-tran & Coin-to-coin & Sender anon. & Recip. anon. & Sent coin & Value of coin & Time of coin & Consumed coin & Value of coin & Time of coin & \\ \hline
Bitcoin & \cmark(IP addr) & \cmark(IP addr) & \xmark & \xmark & \xmark & \cmark\cmark & - & - & - & - & - & - & - \\ \hline
CoinJoin & \cmark(IP addr) & \cmark(IP addr) & \xmark & \cmark(prob.) & \xmark & \cmark\cmark & - & - & - & - & - & - & - \\ \hline
Coinshuffle & \cmark(IP addr) & \cmark(IP addr) & \xmark & \cmark(prob.) & \xmark & \cmark\cmark & - & - & - & - & - & - & - \\ \hline
Zerocoin & \cmark(IP addr) & \cmark(IP addr) & \xmark & \cmark\cmark & \xmark & \cmark\cmark & \cmark\cmark & \cmark\cmark & \cmark\cmark & \cmark\cmark & \cmark\cmark & \cmark\cmark & \cmark\cmark(SLLA-resist., addressless, SCCLA-resist.) \\ \hline
Zerocash & \cmark(IP addr) & \cmark(IP addr) & \xmark & \cmark\cmark & \xmark & \cmark\cmark & \cmark\cmark & \cmark\cmark & \cmark\cmark & \cmark\cmark & \cmark\cmark & \cmark\cmark & \cmark\cmark(SLLA-resist., address encryption, SCCLA-resist.) \\ \hline
CryptoNote & \cmark(IP addr) & \cmark(IP addr) & \xmark & \cmark(prob.) & \xmark & \cmark\cmark & \cmark\cmark & \cmark\cmark & \cmark\cmark & \cmark & \cmark\cmark(Ring CT) & \cmark & \cmark(SLLA-prob.-resist., one-time address, SCCLA-prob.-resist.) \\ \hline
Mimblewimble & \cmark(IP addr) & \cmark(IP addr) & \xmark & \xmark & \cmark\cmark & \xmark & \cmark\cmark & \cmark\cmark & \cmark\cmark & \xmark & \cmark\cmark & \xmark & \cmark(SLLA-unresist., addressless, SCCLA-unresist.) \\ \hline
\end{tabular}
}
\caption{Privacy Evaluation Results of Anonymous Cryptocurrencies}
\label{table:privacy_evaluation_results}
\end{table*}

\subsection{Secret-Sharing Layer (Layer-2)}
\subsubsection{Trustless Security Model}
In the case of anonymous messaging protocols \cite{Unger2015}, the anonymity is evaluated under the trust established between a sender and a recipient, as Pfitzmann states, ``we assume that the attacker is not able to get information on the sender or recipient from the message content,'' in which each user (e.g., a recipient) does not recognize their counterparty (e.g., a sender) as an attacker.
However, since anonymous cryptocurrencies are developed based on blockchain technologies, the anonymity in anonymous cryptocurrencies is expected to be evaluated under the trustless security model.
In a transaction of anonymous cryptocurrencies, each recipient does not necessarily trust even the sender.

\subsubsection{Pseudonymity}
In Layer-2 of each anonymous cryptocurrency, since the same address scheme is used as in Layer-1, we will not consider pseudonymity again here.

\subsubsection{Unlinkability}
The linkability attack models let us clearly understand the behavior of each user from the viewpoint of achieving anonymity. Since anonymous cryptocurrencies are typically based on trustless security models, any counterparty of transactions can be malicious. Therefore, we will propose linkability attack models in Layer-2, in which even a counterparty of a sender or recipient of each coin is an attacker that might link any IOI and sometimes de-anonymize the victim.

Note that, in these attack models, there is a possibility that the attacker has already known victim's PII to some extent, e.g., they have already de-anonymized some amount of pseudonyms by colluding with a linkability broker, or they are an online store and already know the victim's email address. Therefore, any information obtained from the attacks can be instantly linked to the victim's PII. That is the reason why we are focusing on the attack models in Layer-2.

We will here define an attack model in which a recipient of coins is an attacker, and a sender of the coins is a victim.

\begin{definition}
Receiving Chosen-Coin Linkability Attack (RCCLA):\\
An attacker tries to obtain any unknown information about a victim from coins sent from the victim to link with the information they already have on condition that the attacker can receive coins with any parameter from the victim. Note that the information that can be retrieved from the ledger is included in the attacker's initial knowledge.
\end{definition}

In most cases, the online store is a recipient side that exchanges anonymous cryptocurrency with their product or service. Therefore, most online stores can potentially be attackers using RCCLA, and RCCLA is the most crucial attack model in the real world.

However, RCCLA is relatively challenging for the attacker because the coins can be sent to them only if the victim agrees to send them. And then, RCCLA can be prevented by any appropriate trust establishment techniques.

We will next define an attack model in which a sender of coins is an attacker, and a recipient of the coins is a victim.

\begin{definition}
Sending Chosen-Coin Linkability Attack (SCCLA):\\
An attacker tries to obtain any unknown information about a victim from coins sent from the victim to link with the information they already have on condition that the attacker can send coins with any parameter to the victim. Note that the information that can be retrieved from the ledger is included in the attacker's initial knowledge.
\end{definition}

SCCLA is successful only if the victim outputs some amount of information from their wallets, i.e., the victim sends coins to other users. Collusion between the attacker and a third party to which the victim sends coins makes the success rate of the attack higher (Figure \ref{fig:sccla_collusion}).

\begin{figure}
\centering
\includegraphics[width=7.7cm]{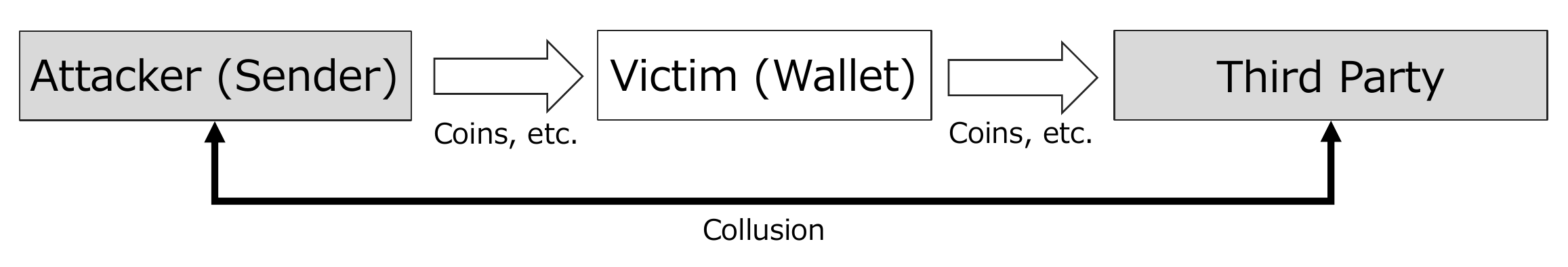}
\caption{SCCLA collusion}
\label{fig:sccla_collusion}
\end{figure}

In the real world, the users in anonymous cryptocurrencies who dedicate themselves to sending coins are unusual. However, SCCLA is relatively easy for an attacker to succeed because the coins can be sent only by the decision of the attacker, and the recipients cannot refuse it in most anonymous cryptocurrencies. Moreover, the victims cannot financially suffer from receiving the coins.

RCCLA-resistance and SCCLA-resistance of each cryptocurrency is scrutinized in Appendix \ref{sec:appendix_unlinkability_evaluation}.

\subsubsection{Anonymity}
Direct de-anonymization can be protected so long as linkability brokers of transport addresses are resistant to cyber-attacks.
However, since linkability brokers are usually TTPs, it depends on conventional cybersecurity techniques employed by them.

Indirect de-anonymization can also be protected so long as all the information in each transaction is encrypted.
However, this depends on the privacy-enhancing techniques for anonymous cryptocurrencies, as shown in \ref{subsub:layer0_unlinkability}.
Therefore, the risk of de-anonymization also depends on privacy-enhancing techniques.

In both cases, once a coin sender or recipient is de-anonymized, their damaged PII depends on unlinkability of the anonymous cryptocurrency.
For example, even in Zerocash, a miner can link two or more encrypted transactions by the medium of a transport address.
Therefore, if the miner has already de-anonymized the sender for some reason, they can link activities of sending coins with the sender's PII by the medium of the transport address, e.g., whether the sender sends a particular transaction or not, or when the sender sends the transaction.

\section{Conclusion}
\label{sec:conclusion}

In this paper, we first proposed an architecture model for anonymous cryptocurrency with three layers.
Next, we studied the definitions of fundamental privacy properties (Pfitzmann's anonymity, unlinkability, and pseudonymity), and comprehensively analyzed each privacy property for each architecture layer in anonymous cryptocurrencies that we defined in this paper.
In this analysis, we proposed four linkability attack models (SLLA, TLLA, RCCLA, SCCLA), with which we cover most of the privacy vulnerabilities in anonymous cryptocurrencies.

Table \ref{table:privacy_evaluation_results} is our privacy evaluation results of Bitcoin and six representative anonymous cryptocurrencies proposed so far (CoinJoin, Coinshuffle, Zerocoin, Zerocash, CryptoNote, Mimblewimble).

First, we reconfirmed that most of the representative anonymous cryptocurrencies achieve higher sender and recipient anonymity than Bitcoin, due to their anonymization schemes (Layer-1 anonymity).
Although allowing attackers to construct the transaction graph is one of the Mimblewimble's shortcomings as shown in \cite{Bogatyy2019}, that does not instantly result in de-anonymization of each user. The reason is that its characteristic of strong anonymity (unnecessity of using addresses) cancels the linkability problem.

Second, we reconfirmed that most of the anonymous cryptocurrencies evaluated in this paper have a common issue regarding the TLLA-resistance (Layer-0 unlinkability). Since a coin sender has to send a transaction message to a miner using their transport address, a miner can associate the transaction information to the transport address. Therefore, anonymization schemes in Layer-0 are required for each anonymous cryptocurrency.

Finally, we found that the unlinkability scheme in Zerocoin and Zerocash is most sophisticated in anonymous cryptocurrencies evaluated in this paper because a coin sender cannot calculate the serial number of the coin, which will be later used when the coin recipient consumes; therefore, the attacker in SCCLA cannot track the coin consumption by the recipient (Layer-2 unlinkability).

\bibliographystyle{IEEEtran}
\bibliography{IEEEabrv,../papers/blockchain,../papers/privacy,../papers/cryptography}

\appendices

\section{Unlinkability Evaluation using Proposed Attack Models}
\label{sec:appendix_unlinkability_evaluation}

In this appendix, we scrutinize SLLA-resistance, RCCLA-resistance, and SCCLA-resistance of three leading cryptography-based anonymous cryptocurrencies (Zerocash \cite{BenSasson2014}, CryptoNote \cite{Saberhagen2013}, and Mimblewimble \cite{Poelstra2016}) to show how to apply each linkability attack model we proposed in this paper.

\subsection{Zerocash}
\subsubsection{Anonymization Scheme}
In Zerocash \cite{BenSasson2014}, all coins are encrypted in the shared ledger. The seemingly unrelated two identifiers, a coin commitment $cm$ and a serial number $sn$ (Figure \ref{fig:zerocash_coin_structure}), are used to identify each coin produced by the previous transaction sender. A transaction sender uses $cm$($cm^{new}$) as an identifier of the coin (called $c^{new}$ from the viewpoint of the transaction sender), while the corresponding recipient uses $sn$($sn^{old}$) as another identifier of the same coin (called $c^{old}$ from the viewpoint of transaction recipient). Consequently, this makes it difficult for most Zerocash users, including even the sender of $c^{new}$, to link the two identifiers of the same coin recorded in the shared ledger (only the recipient of $c^{old}$ can link those).

\begin{figure}
\centering
\includegraphics[width=7.7cm]{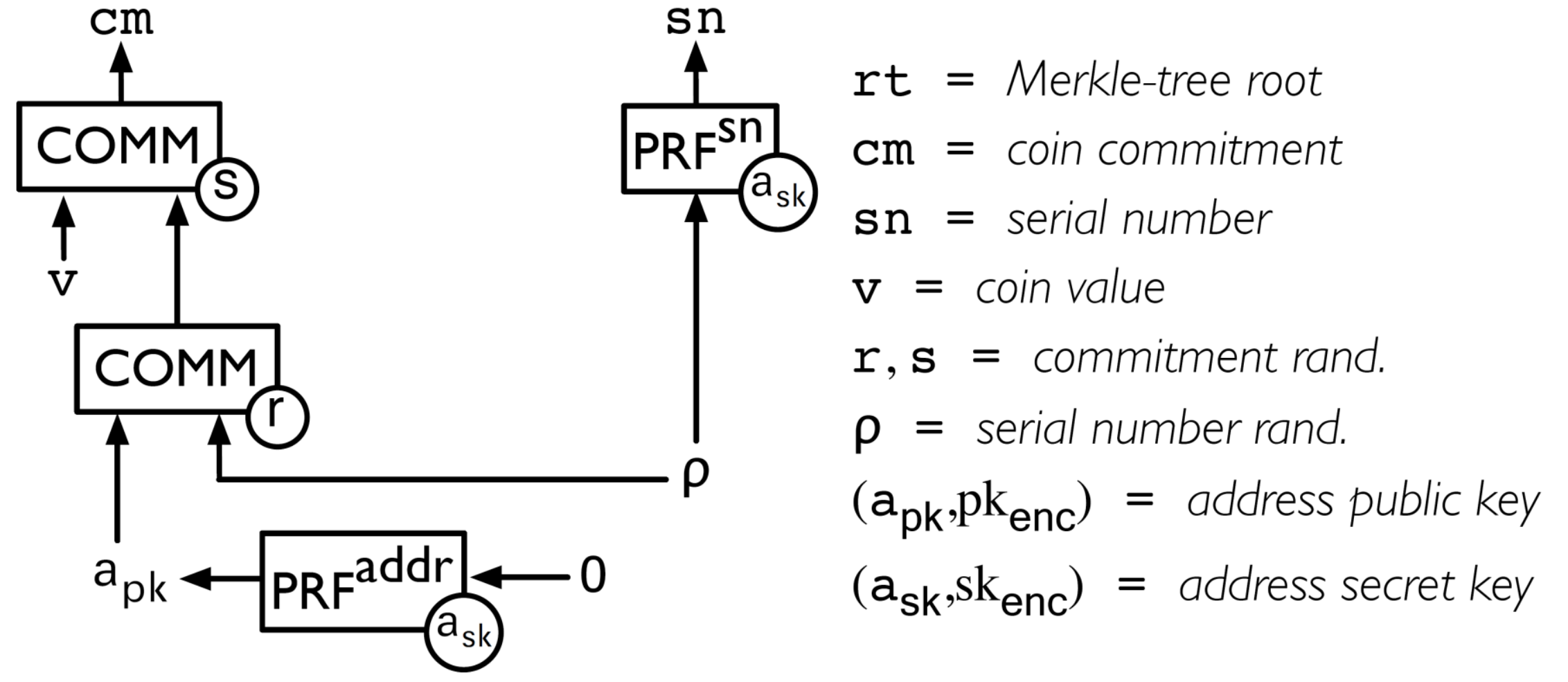}
\caption{Zerocash Coin Structure (This figure is cited from \cite{BenSasson2014})}
\label{fig:zerocash_coin_structure}
\end{figure}

However, people cannot use it in that shape because the recipient cannot verify whether the coins sent from the sender is valid (i.e., the shared secrets sent from the sender are well-formed). Therefore, the sender also includes a zero-knowledge proof $\pi_{pour}$ in each transaction to prove that the identifiers and the secret parameters of each coin are consistent.

\begin{itemize}
\item $cm^{old}$ is properly calculated from the secret information of the coin $c^{old}$ held by the sender (note that $cm^{old}$ is also secret information to the sender.)
\item $cm^{new}$ is properly calculated from the secret information of the coin $c^{new}$ produced and sent to the recipient by the sender.
\item The public key $a_{pk}$ used as an address when the sender has received $cm^{old}$ is properly generated from the sender's secret key $a_{sk}$ (the proof of the proper recipient $a_{pk}$ of $cm^{old}$).
\item $sn^{old}$ is properly calculated from the secret information of the coin $c^{old}$ held by the sender and the sender's secret key $a_{sk}$.
\item The secret information $cm^{old}$ is included in the $cm$ list recorded in the shared ledger (the proof of the fact that $cm^{old}$ had already existed when the sender received the coin $c^{old}$).
\item The total amount of old coins ($v_1^{old} + v_2^{old} + ...$) included in the transaction input is equal to the total amount of new coins ($v_1^{new} + v_2^{new} + ...$) included in the transaction output (note that amount of each coin is secret information).
\end{itemize}

Preventing double-spending of the coin $c^{old}$ does not need to be proved using zero-knowledge proofs. Verifying that $sn^{old}$ has not yet appeared in each of the transaction inputs recorded in the shared ledger is sufficient.

\subsubsection{SLLA-Resistance}
In Zerocash, $cm$ and $sn$ are the only values recorded in the shared ledger. Since $cm$ is generated by a commitment scheme and $sn$ is a returned value by a pseudo-random function, all those values are recognized as random by any third party, and it is difficult for them to link any two values recorded in the shared ledger.

Therefore, Zerocash is SLLA-resistant.

\subsubsection{RCCLA-Resistance}
When a sender $a_{pk}$ (a victim in RCCLA) sends a new coin $c^{new}$ to a recipient $b_{pk}$ (an attacker in RCCLA) through a transaction $tx_{new}$, the sender discloses to the recipient the serial number $sn^{old}$ of the consumed coin $c^{old}$ and the shared parameters of the sent coin $c^{new}$, e.g., $b_{pk}$, $v^{new}$, $\rho^{new}$, $r^{new}$, $s^{new}$, and $cm^{new}$.

The following is our logical verification on whether the recipient $b_{pk}$ of the coin $c^{new}$ can obtain any unknown information about the sender $a_{pk}$ by the medium of each of the shared parameters of the coin $c^{new}$. Note that since RCCLA is an attack model under the condition that the attacker can receive coins with any parameter from the victim, we assume that the attacker can make the recipient send coins with maliciously induced parameters as needed.
\\

\begin{enumerate}
\renewcommand{\labelenumi}{\arabic{enumi})}
\item Serial Number: $sn^{old}$ \\
Since the recipient needs the secret key $z_{sk}$ of the sender $z_{pk}$ of the coin $c^{old}$ in addition to the parameters of the coin $c^{old}$ to calculate the coin commitment $cm^{old}$, the recipient cannot link $sn^{old}$ with $cm^{old}$.

Next, the serial number $sn^{old}$ cannot appear in any other transactions than $tx_{new}$, because the secondary use of the serial number is regarded as double-spending.

Therefore, the attacker $b_{pk}$ cannot obtain any unknown information about the sender $a_{pk}$ by the medium of the serial number $sn^{old}$.

\item Recipient's Address: $b_{pk}$ \\
The recipient's address $b_{pk}$ is their address of the attacker.

Therefore, the attacker $b_{pk}$ cannot obtain any unknown information about the sender $a_{pk}$ by the medium of the recipient's address $b_{pk}$.

\item Amount of Coin: $v^{new}$ \\
Since the new coins $(c^{new}_{i})$ are always produced by the sender $a_{pk}$ every time the old coins $(c^{old}_{i})$ are consumed, where the old coins $(c^{old}_{i})$ are merged or split into the new coins $(c^{new}_{i})$, the same amount $v^{old}$ of one of the old coins $(c^{old}_{i})$ is not used as an amount $v^{new}$ of one of the new coins $(c^{new}_{i})$ unless the sender intentionally reuses it.

Therefore, the attacker $b_{pk}$ cannot obtain any unknown information about the sender $a_{pk}$ by the medium of the amount of coin $v^{new}$.

\item Random Parameters: $\rho^{new}$, $r^{new}$, and $s^{new}$ \\
Since a new coin $c^{new}$ is always produced by the sender $a_{pk}$ every time the old coin $c^{old}$ is consumed, the same random parameters $\rho^{old}$, $r^{old}$, and $s^{old}$ of the old coin $c^{old}$ are not used as random parameters $\rho^{new}$, $r^{new}$, and $s^{new}$ of the new coins $c^{new}$ unless the sender intentionally reuses it.

Therefore, the attacker $b_{pk}$ cannot obtain any unknown information about the sender $a_{pk}$ by the medium of the random parameters $\rho^{new}$, $r^{new}$, and $s^{new}$.

\item Coin Commitment: $cm^{new}$ \\
The recipient $b_{pk}$ can calculate the serial number $sn^{new}$ of the coin $c^{new}$ using their secret key $b_{sk}$, identify the corresponding coin commitment $cm^{new}$ in a transaction $tx_{new}$ on the shared ledger, and hence link $sn^{new}$ with $cm^{new}$ and $sn^{old}$. However, the recipient cannot link $cm^{new}$ and $sn^{old}$ with anymore.

Next, the coin commitment $cm^{new}$ cannot appear in any other transactions than $tx_{new}$, because CMList is not allowed to include the same two or more coin commitments at the same time.

Therefore, the attacker $b_{pk}$ cannot obtain any unknown information about the sender $a_{pk}$ by the medium of the coin commitment $cm^{new}$.
\end{enumerate}

In the end, the attacker $b_{pk}$ in RCCLA cannot obtain any unknown information about a victim $a_{pk}$ from the coins sent from the victim to link with the information they already have.

Therefore, Zerocash is RCCLA-resistant.

\subsubsection{SCCLA-Resistance}
When a sender $a_{pk}$ (an attacker in SCCLA) sends a $c^{new}$ to a recipient $b_{pk}$ (a victim in SCCLA) through a transaction $tx_{new}$, the sender discloses to the recipient the $sn^{old}$ of the consumed $c^{old}$ and the shared parameters of the sent $c^{new}$, e.g., $b_{pk}$, $v^{new}$, $\rho^{new}$, $r^{new}$, $s^{new}$, and $cm^{new}$.

The following is our logical verification on whether the sender $a_{pk}$ of $c^{new}$ can obtain any unknown information about the recipient $b_{pk}$ from the $c^{next}$ sent to a colluder $c_{pk}$ by the recipient by the medium of each of the shared parameters of $c^{new}$. Note that since SCCLA is an attack model under the condition that the attacker can send coins with any parameter to the victim, we assume that the attacker can send coins with maliciously manipulated parameters as needed.
\\

\begin{enumerate}
\renewcommand{\labelenumi}{\arabic{enumi})}
\item Serial Number: $sn^{old}$ \\
The recipient $b_{pk}$ does not use the serial number $sn^{old}$ to consume the coin $c^{new}$. Furthermore, since even the sender cannot calculate the serial number $sn^{new}$ of the sent coin $c^{new}$ because anyone who wants to calculate $sn^{new}$ needs the secret key $b_{sk}$ of the recipient $b_{pk}$, the sender cannot link $sn^{old}$ and $sn^{new}$.

Therefore, the attacker $a_{pk}$ cannot obtain any unknown information about the recipient $b_{pk}$ by the medium of the serial number $sn^{old}$ and $sn^{new}$.

\item Recipient's Address: $b_{pk}$ \\
Since the recipient is not explicitly notified of any address in the Zerocash protocol, the recipient's address $b_{pk}$ is not transmitted to any third party (who may collude with the sender) $c_{pk}$ unless the recipient intentionally notifies the third party of its address outside the protocol.

Therefore, the recipient's address $b_{pk}$ is not spread to anyone, and the attacker $a_{pk}$ cannot obtain any unknown information about the recipient $b_{pk}$ by the medium of the recipient's address $b_{pk}$.

\item Amount of Coin: $v^{new}$ \\
The sender $a_{pk}$ can set an amount of a coin $v^{new}$ to any intentional value, e.g., 0.07007007ZEC, to track the recipient's behavior. However, since the next coins $(c^{next}_{i})$ are always produced by the recipient $b_{pk}$ every time the new coins $(c^{new}_{i})$ are consumed, where the new coins $(c^{new}_{i})$ are merged or split into the next coins $(c^{next}_{i})$, the same amount $v^{new}$ of one of the new coins $(c^{new}_{i})$ is not used as an amount $v^{next}$ of one of the next coins $(c^{next}_{i})$ unless the recipient intentionally reuses it.

Therefore, the amount of coin $v^{new}$ is not spread to anyone other than the recipient $b_{pk}$, and the attacker $a_{pk}$ cannot obtain any unknown information about the recipient $b_{pk}$ by the medium of the amount of coin $v^{new}$.

\item Random Parameters: $\rho^{new}$, $r^{new}$, and $s^{new}$ \\
Since a next coin $c^{next}$ is always produced by the recipient $b_{pk}$ every time the new coin $c^{new}$ is consumed, the same random parameters $\rho^{new}$, $r^{new}$, and $s^{new}$ of the new coin $c^{new}$ are not used as random parameters $\rho^{next}$, $r^{next}$, and $s^{next}$ of the next coins $c^{next}$ unless the recipient intentionally reuses it.

Therefore, the random parameters $\rho^{new}$, $r^{new}$, and $s^{new}$ are not spread to anyone other than the recipient $b_{pk}$, and the attacker $a_{pk}$ cannot obtain any unknown information about the recipient $b_{pk}$ by the medium of the random parameters $\rho^{new}$, $r^{new}$, and $s^{new}$.

\item Coin Commitment: $cm^{new}$ \\
Since the coin commitment $cm^{next}$ that the recipient $b_{pk}$ will produce is calculated from the random parameters $\rho^{next}$, $r^{next}$, and $s^{next}$, etc., $cm^{next}$ cannot be the same value as $cm^{new}$.

Therefore, the coin commitment $cm^{new}$ is not spread to anyone other than the recipient $b_{pk}$, and the attacker $a_{pk}$ cannot obtain any unknown information about the recipient $b_{pk}$ by the medium of the coin commitment $cm^{new}$.
\end{enumerate}

In the end, the attacker $a_{pk}$ in SCCLA and their colluder $c_{pk}$ cannot obtain any unknown information about a victim $b_{pk}$ from the coins sent to the victim to link with the information they already have.

Therefore, Zerocash is SCCLA-resistant.

\subsection{CryptoNote}
\subsubsection{Anonymization Scheme}
CryptoNote \cite{Saberhagen2013}, which is mainly designed based on Bitcoin, enhances its anonymity using schemes named stealth address and one-time ring signature.

In CryptoNote, the transaction sender first provides a one-time randomized recipient address (stealth address) generated from a one-time random value $r$ sampled by the sender and a pair of public keys $(A, B)$ generated by the recipient. Therefore, the recipient's stealth address varies every time, even if the sender sends two or more transactions to the same recipient, thus making CryptoNote achieve recipient anonymity.

However, the stealth address scheme is not sufficient because it still allows de-anonymizing attackers to construct a transaction graph and analyze it. Therefore, CryptoNote provides another anonymization scheme (one-time ring signature), where the signer who consumes each coin bundles their public key and $(n-1)$ other user's dummy public keys to generate a non-interactive zero-knowledge proof that suggests that the coin is signed (or consumed) by the secret key corresponding to one of the bundled public keys (although which key pair is used is not disclosed). Due to this scheme, the de-anonymization attack using the transaction graph is challenging in CryptoNote. However, the success rate of the de-anonymization attack cannot be reduced to below $1/n$ ($n$ is the number of bundled public keys). The degree of sender anonymity depends on $n$. Note that the one-time ring signature scheme has almost the same effect as Bitcoin's mixing techniques \cite{Maxwell2013}\cite{Ruffing2017}.

Whether a sender uses a different stealth address (i.e., different $r$) for each transaction depends on the sender. However, the verifying algorithm of double-spending in one-time ring signature does not allow the recipient to consume both coins signed with the same stealth address.

Note that the amount of coin in CryptoNote has already been encrypted by the method named Ring Confidential Transactions (Ring CT) \cite{Noether2016}.

\subsubsection{SLLA-Resistance}
In the case of CryptoNote, transactions, coins (UTXOs), a sender of each coin, and a recipient of each coin are entirely disclosed to an unspecified large number of users on a shared ledger.

However, since each address of the coin recipients and thus senders is randomized by stealth address, any two coins are unlinkable by the medium of the sender's nor recipient's address, which means that the coins achieve sender and recipient anonymity (\ref{sub:relationship_anonymity_unlinkability}).

On the other hand, although one-time ring signature scrambles the signer who consumes the coin, linking the received coin with a sent coin is possible (success rate is $1/n$). Therefore, the attackers of SLLA can construct a probabilistic transaction graph.

Therefore, CryptoNote is not SLLA-resistant (but weakly SLLA-resistant).

\subsubsection{RCCLA-Resistance}
When a sender $a_{pk}$ (a victim in RCCLA) sends a new coin $c^{new}$ to a recipient $b_{pk}$ (an attacker in RCCLA) through a transaction $tx_{new}$, the sender discloses to the recipient the destination key $P$, the amount of coin $v^{new}$, and the tx public key $R$.

The following is our logical verification whether the recipient $b_{pk}$ of the coin $c^{new}$ can obtain any unknown information about the sender $a_{pk}$ by the medium of each of the shared parameters of the coin $c^{new}$.
\\

\begin{enumerate}
\renewcommand{\labelenumi}{\arabic{enumi})}
\item Destination Key: $P$ \\
The destination key $P$ of the coin $c^{new}$ is a stealth address of the recipient $b_{pk}$.
Since the stealth address is exclusively generated for each coin, the recipient cannot find it in any other coins.

Therefore, the attacker $b_{pk}$ cannot obtain any unknown information about the sender $a_{pk}$ by the medium of the destination key $P$.

\item Amount of Coin: $v^{new}$ \\
Since the new coins $c^{new}_{i}$ are always produced by the sender $a_{pk}$ every time the old coins $c^{old}_{i}$ are consumed, where the old coins $c^{old}_{i}$ are merged or split into the new coins $c^{new}_{i}$, the same amount $v^{old}$ of one of the old coins $c^{old}_{i}$ is not used as an amount $v^{new}$ of one of the new coins $c^{new}_{i}$ unless the sender intentionally reuses it.

Therefore, the attacker $b_{pk}$ cannot obtain any unknown information about the sender $a_{pk}$ by the medium of the amount of coin $v^{new}$.

\item Tx Public Key: $R$ \\
The tx public key $R$ of the transaction $tx^{new}$ is a corresponding public key to the one-time random value $r$. If using the same $r$ in different transactions is prohibited, the recipient cannot find it in any other transaction.

Therefore, the attacker $b_{pk}$ cannot obtain any unknown information about the sender $a_{pk}$ by the medium of the tx public key $R$.

\end{enumerate}

In the end, the attacker $b_{pk}$ in RCCLA cannot obtain any unknown information about a victim $a_{pk}$ from the coins sent from the victim to link with the information they already have.

Therefore, CryptoNote is RCCLA-resistant.

\subsection{Mimblewimble}
\subsubsection{Anonymization Scheme}
\label{subsub_mimblewimble_abstract}
In Mimblewimble \cite{Poelstra2016}, the amount of each coin $v$ is hidden in Pedersen commitment \cite{Pedersen1992}, $p=rG+vH$, which is used as a token of a coin. If someone wants to send the coins to another, the sender sends the commitments with their secrets $r_i$ and $v_i$ to the recipient, the recipient constructs a transaction named Confidential Transactions \cite{Maxwell2016}, including the sent old commitments and the new commitments that the recipient newly creates at that moment. The recipient sends the transaction to the shared ledger. In Mimblewimble, only the recipient of the old commitments consumes the coins because they know the secrets of the old commitments. Thus, they can construct a transaction of Confidential Transactions, in which the sum of all the commitments is zero.

Although the sender still knows the secrets of the old commitments, they cannot consume the old coins after the recipient sends the transaction to the shared ledger because the miners do not allow anyone to double-spend. On the other hand, since only the recipient knows the secrets of the new commitments, the new coins cannot be consumed by anyone else. Additionally, in the transaction of Confidential Transactions, the total amount of new coins is made to equal the total amount of the old coins. Therefore, the ownership of the coins is correctly and safely transferred from the sender to the recipient.

In summary, the possession of a coin is determined by holding the secrets of the commitment, and the transfer of the coin is achieved by sending the secrets of the commitment. Therefore, address information for identifying the users is not required; thus, it is not included in each coin in Mimblewimble. That is, we do not have to anonymize the address in Mimblewimble.

\subsubsection{SLLA-Resistance}
In the case of Mimblewimble, the transactions of Confidential Transactions, including the Pedersen commitments, are disclosed to an unspecified large number of users on a shared ledger.

Since the amount of each coin is hidden in a commitment and each address of the coin senders and recipients is not included in the coin, any two coins are unlinkable by the medium of sender's address or recipient's address, which means that the coins achieve sender and recipient anonymity (\ref{sub:relationship_anonymity_unlinkability}). Having no address brings substantial unlinkability by the medium of any address, avoiding de-anonymization of anyone.

However, since the identifier of the sent coin and that of the received coin is the same, anyone can construct the transaction graph.

Therefore, Mimblewimble is not SLLA-resistant.

\subsubsection{RCCLA-Resistance}
When a sender $a_{pk}$ (a victim in RCCLA) sends an old coin $c^{old}$ to a recipient $b_{pk}$ (an attacker in RCCLA), the sender discloses the old commitments with their secrets $r$ and $v$ to the recipient.

However, since the recipient just consumes the old commitments, the recipient is no longer interested in the commitment. Although the recipient can still construct the transaction graph, they cannot obtain unknown information more than an attacker in SLLA about a victim $a_{pk}$ knows from the coins sent from the victim to link with the information they already have.

Therefore, Mimblewimble is RCCLA-resistant.

\subsubsection{SCCLA-Resistance}
When a sender $a_{pk}$ (an attacker in SCCLA) sends an old coin $c^{old}$ to a recipient $b_{pk}$ (a victim in SCCLA), the sender discloses the old commitments with their secrets $r$ and $v$ to the recipient.

Since the sender can track the new coin $c^{new}$ by the medium of the old commitments, they will be able to know when the new coin will be consumed.

In the end, the attacker $a_{pk}$ in SCCLA can obtain some unknown information about a victim $b_{pk}$ by the medium of the coins sent to the victim to link with the information they already have.

Therefore, Mimblewimble is not SCCLA-resistant.

\end{document}